\documentclass[12pt,journal,draftclsnofoot,onecolumn]{IEEEtran}
\usepackage{flushend,cuted}

\usepackage{color}
\usepackage{cite}
\usepackage[cmex10]{amsmath}

\usepackage{tabularx,colortbl}
\definecolor{Gray}{gray}{0.9}

\usepackage{amsmath,amssymb,amsthm,mathrsfs,dsfont}

\usepackage{amsmath}
\usepackage{mathtools}
\usepackage{amssymb}

\usepackage{booktabs}
\usepackage{threeparttable}
\usepackage{multirow}

\usepackage{algorithm}
\usepackage{algpseudocode}

\usepackage{subfigure}
\usepackage{stfloats}
\usepackage[pdftex]{graphicx}

\newtheorem{theorem}{Theorem}

\newtheorem{proposition}{Proposition}

\IEEEoverridecommandlockouts

\begin{document}

\title{\huge Distributed Power Control in Interference Channels with QoS Constraints and RF Energy Harvesting: A Game-Theoretic Approach
\thanks{The authors are with the School of Electrical and Information Engineering, the University of Sydney, NSW, 2006, Australia (e-mail:~\{he.chen,~yuanye.ma,~zihuai.lin,~yonghui.li,~branka.vucetic\}@sydney.edu.au). H. Chen and Y. Ma contributed equally to the paper.
}
}
\author{
He~(Henry)~Chen,~
Yuanye Ma,~
Zihuai~Lin,~
Yonghui~Li,~
Branka~Vucetic
}
\maketitle

\begin{abstract}
This paper develops a new distributed power control scheme for a power splitting-based interference channel (IFC) with simultaneous wireless information and power transfer (SWIPT). The considered IFC consists of multiple source-destination pairs. Each destination splits its received signal into two parts for information decoding and energy harvesting (EH), respectively. Each pair adjusts its transmit power and power splitting ratio to meet both the signal-to-interference-plus-noise ratio (SINR) and EH constraints at its corresponding destination. To characterize rational behaviors of source-destination pairs, we formulate a non-cooperative game for the considered system, where each pair is modeled as a strategic player who aims to minimize its own transmit power under both SINR and EH constraints at the destination. We derive a \emph{sufficient and necessary} condition for the existence and uniqueness of the Nash equilibrium (NE) of the formulated game. The best response strategy of each player is derived and then the NE can be achieved iteratively. Numerical results show that the proposed game-theoretic approach can achieve a near-optimal performance under various SINR and EH constraints.
\end{abstract}
\begin{IEEEkeywords}
Simultaneous wireless information and power transfer, interference channel, distributed power control, game theory, Nash equilibrium.
\end{IEEEkeywords}

\IEEEpeerreviewmaketitle

\section{Introduction}

Wireless energy transfer (WET) technology has drawn significant interest over the past decade because of its potentially wide applications in our daily life. Recently, the recent emerging WET techniques have enabled wireless devices to harvest energy from ambient/dedicated radio frequency (RF) signals~\cite{Lu_arxiv_2014_wireless}. On the other side, it is well known that RF signals are also used to carry information in wireless communications. As a result, simultaneous wireless information and power transfer (SWIPT) \cite{Krikidis_2014_simultaneous} has recently been proposed to realize the dual utilization of RF signals for joint information and energy transfer at the same time. Such dual utilization of RF signals in SWIPT leads to different system design in various setups and applications. For example, in an interference channel (IFC) with SWIPT, the cross-link interference is still harmful to the information decoding (ID) at the receiver side, but it becomes beneficial when we pay more attention to the energy harvesting (EH) aspect.

There have been several papers in the open literature that focus on the design of SWIPT in IFCs \cite{shen2013wireless,park2013joint,parkKuserjoint,timotheou2013beamforming,shi2014joint}. Specifically, \cite{shen2013wireless} considered SWIPT in a multiple-input single-output (MISO) IFC, where the weighted sum-rate was maximized subject to individual EH constraints and transmission power constraints. In \cite{park2013joint}, all possible transmission strategies with different combinations of information decoding and energy harvesting at the receiver side were investigated and compared in a two-user multiple-input multiple-output (MIMO) IFC, which was subsequently extended to the general $K$-user case in \cite{parkKuserjoint}. \cite{timotheou2013beamforming} studied the joint beamforming and power splitting problem in a MISO IFC, where the total transmit power of all transmitters was minimized under both rate and EH constraints by employing the semidefinite programming (SDP) method. Different from \cite{timotheou2013beamforming}, a second-order cone programming (SOCP) relaxation-based approach was developed in \cite{shi2014joint} as an alternative solution to resolve the same total power minimization problem in a decentralized manner. In all aforementioned papers that designed SWIPT schemes in IFCs, it is assumed that all source-destination pairs cooperate to achieve the optimal network-wide performance (e.g., maximizing the sum-rate/minimizing the total transmit power of all pairs). However, in many practical scenarios, source-destination pairs may be rational such that they only care about their own performance instead of the overall one (see \cite{larsson2009game} and references therein). To the best of our knowledge, there has been no work that designs SWIPT for the IFC with self-interested source-destination pairs in the open literature. This gap actually motivates our paper. It is worth pointing out that game theory has actually been applied to investigate different setups
of communication networks involving RF energy harvesting in \cite{Chen_TWC_2015_Dis,Ma_PIMRC_2015_Dis,Chen_ICASSP_2015_A,Ma_Tcom_2015_Dis}. But none of them studied the concerned scenario of IFCs with SWIPT and QoS constraints.

In this paper, we develop a game-theoretic framework for the distributed power control in a power-splitting based IFC with SWIPT. The main \emph{contributions} of this paper are summarized as follows: \textbf{(1)} we formulate a non-cooperative game for the considered IFC-SWIPT system, where each source-destination pair is modeled as a strategic player who aims to minimize the source transmit power while satisfying both signal-to-interference-plus-noise ratio (SINR) and EH constraints at the destination; \textbf{(2)} we derive a sufficient and necessary condition to guarantee the existence and uniqueness for the Nash equilibrium (NE) of the formulated game; \textbf{(3)} the best response strategy for each pair is derived and then the NE is achieved in a distributed manner; \textbf{(4)} numerical results validate the theoretical analysis and show that the proposed game-theoretic approach can achieve a near-optimal network-wide performance.

\section{System Model and Game Formulation}
In this section, we first describe the system model and then formulate a non-cooperative game for the considered network.

\subsection{System Model}

In this paper we consider {a single-input single-output (SISO)\footnote{It would be more interesting to study a general multiple-input multiple-output (MIMO) IFC with SWIPT. However, this would need to develop a totally new game-theoretic framework for the joint design of transmit power, transmit and receive beamforming vectors and power splitting ratio at each source-destination pair, which may constitute another full paper and thus have been left as our future~work.}} IFC consisting of $N$ source-destination pairs. We use $\mathcal{N} =\{1, \cdots, n, \cdots, N \}$ to denote the index set of source-destination pairs, in which the $n$th pair consists of the $n$th source and the $n$th destination. All pairs are assumed to share the same frequency band and thus they interfere with each other. We assume that all sources and destinations are equipped with one antenna and operate in a half-duplex mode. All destinations are equipped with a power splitting \cite{Krikidis_2014_simultaneous} devices such that they are able to decode the information as well as harvest energy from the received signal at the same time. Besides, we consider that the links between all nodes experience the slow and frequency-flat fading.

Let $h_{nn}$ and $h_{mn}$ denote the channel gain from the $n$th source to the $n$th destination and that from the $m$th source to the $n$th destination, respectively. At the $n$th destination, the received signal before power splitting can be expressed as
\begin{equation}
\begin{split}
y_n = h_{nn} \sqrt{p_n} x_n +  \sum_{m\in \mathcal{N}/\{n\}} h_{mn} \sqrt{ p_m}x_m + z_n,
\end{split}
\end{equation}
where $p_n$ and $p_m$ are transmit powers of the $n$th source and $m$th source, respectively. $x_n$ and $x_m$ denote the unit-energy symbols transmitted by the $n$th source and $m$th source. $z_n \sim \mathcal{CN} (0, \delta_n^2)$ is the additive noise introduced by the receiver antenna at the $n$th destination.

The $n$th destination splits the received signal into two streams with a power splitting ratio $\alpha_n$. The fraction $\sqrt{\alpha_n}$ of the received signal is used for EH, while the remaining $\sqrt{1-\alpha_n}$ fraction is passed to the ID unit.

The ID unit at each receiver will introduce an additional baseband noise to the signal stream passed to the ID circuit. We assume that this additional baseband noise is an additive Gaussian random variable with zero mean and variance $\sigma_n^2$ and it should be independent of the antenna noise $z_n$. Accordingly, we can express the respective harvested energy and received SINR at the $n$th destination by
\begin{equation}\label{eq:def_harvested_energy}
\begin{split}
E_n\left(p_n, \alpha_n;  \boldsymbol{p}_{-n} \right) = \eta \alpha_n  \sum_{m\in\mathcal{N}}  p_m G_{mn},
\end{split}
\end{equation}
\begin{equation}\label{eq:def_SINR}
\begin{split}
\mathrm{SINR}_n\left(p_n, \alpha_n; \boldsymbol{p}_{-n}\right) =   \frac{\left(1- \alpha_n\right) p_n G_{nn}}{\left(1-\alpha_n\right)  \left ( \sum_{m \in \mathcal{N}/\{n\}}p_m G_{mn} + \delta_n^2 \right)  + \sigma_n^2},
\end{split}
\end{equation}
where $\boldsymbol{p}_{-n} \triangleq [p_1, \cdots, p_{n -1}, p_{n+1}, \cdots, p_N]^T$, $G_{nn} \triangleq |h_{nn}|^2$, $G_{mn} \triangleq |h_{mn}|^2$, and $0 < \eta <1$ is the energy conversion efficiency. The notation $\left(p_n, \alpha_n;  \boldsymbol{p}_{-n} \right)$ indicates that $p_n$ and $\alpha_n$ are variables given $\boldsymbol{p}_{-n}$.
Note that in (\ref{eq:def_harvested_energy}), we ignore the amount of energy harvested from the antenna noise since it is normally below the sensitivity of the energy harvesting device in practice \cite{zhou2013wireless}.

\subsection{Game Formulation}

We follow \cite{timotheou2013beamforming,shi2014joint} and assume that each destination is subject to strict QoS and EH constraints. The QoS constraint requires that the received SINR at the $n$th destination, i.e., $\mathrm{SINR}_n$ given in (\ref{eq:def_SINR}), should be no less than a predefined threshold $\gamma_n$, while the EH constraint imposes the condition that the value of $E_n$ defined in (\ref{eq:def_harvested_energy}) should be larger than or equal to the energy threshold $\mathcal{E}_n$.
In contrast to \cite{timotheou2013beamforming,shi2014joint}, which assume that all source-destination pairs are cooperative to minimize the network total transmit power, we consider an alternative non-cooperative scenario in which the source-destination pairs are all rational and self-interested such that they only want to minimize their respective transmit powers under their individual SINR and EH constraints. In this case, the results presented in \cite{timotheou2013beamforming,shi2014joint} by minimizing the network total transmit power may no longer be applicable for the considered non-cooperative scenario.
In this regard, we model the considered system by the well-established game theory \cite{han2012game}. Specifically, we formulate the following non-cooperative game:
\begin{itemize}
  \item \textit{Players}: The $N$ source-destination pairs.
  \item \textit{Actions}: Each pair determines its source transmit power and the power splitting ratio at the destination, i.e., $\left(p_n, \alpha_n\right)$, to minimize its transmit power under the SINR and EH constraints at the destination.
  \item \textit{Utilities}: The source transmit powers $p_n$.
\end{itemize}

We notice that each player's strategy $\left(p_n, \alpha_n\right)$ only depends on the transmit powers of the others $\boldsymbol{p}_{-n}$, as the power splitting ratio of each player only operates at the destination side. Thus, given the transmit powers of the others $\boldsymbol{p}_{-n}$, the best response strategy of the $n$th pair (player) is the solution to the following optimization problem,
\begin{equation}\label{Prob.1}
\begin{split}
& \min_{\{p_n,\alpha_n\}}~ p_n,\\
\mathrm{s.t.}~& \mathrm{C1}: 0 \leq \alpha_n < 1,\\
& \mathrm{C2}: p_n \geq 0,\\
&  \mathrm{C3}: \frac{\left(1- \alpha_n\right) p_n G_{nn}}{\left(1-\alpha_n\right)  \left( \sum_{m \in \mathcal{N}/\{n\}} p_m G_{mn} + \delta_n^2\right) + \sigma_n^2} \geq  \gamma_n,\\
& \mathrm{C4}: \eta \alpha_n  \sum_{m\in\mathcal{N}}  p_m G_{mn} \geq \mathcal{E}_n.
\end{split}
\end{equation}

We denote by $\left(p_n^\star, \alpha_n^\star\right)$ the optimal solution to the optimization problem (\ref{Prob.1}). It is straightforward to verify that (\ref{Prob.1}) always has a feasible solution \cite{Shi_TWC_2014_Joint}. Considering that $p_n$ and $\alpha_n$ should be jointly optimized and they are all dependent on $\boldsymbol{p}_{-n}$, we can rewrite $\alpha_n$ as a function of $p_n$ and $\boldsymbol{p}_{-n}$, i.e., $\alpha_n = f_n \left(p_n, \boldsymbol{p}_{-n}\right)$. Then, we denote by $\mathcal{P}_n\left(\boldsymbol{p}_{-n}\right)$ the feasible power policies of the $n$th pair given the others' power strategies, which can be expressed as
\begin{equation}
\begin{split}
\mathcal{P}_n \left(\boldsymbol{p}_{-n}\right) \triangleq   \left \{  p_n > 0 :\right. & \mathrm{SINR}_n \left( p_n,  f_n \left(p_n, \boldsymbol{p}_{-n}\right); \boldsymbol{p}_{-n} \right) \geq \gamma_n, \\
& \left. E_n \left(p_n,  f_n \left(p_n, \boldsymbol{p}_{-n}\right);  \boldsymbol{p}_{-n} \right) \geq \mathcal{E}_n \right\}.
\end{split}
\end{equation}

Now, we define $\boldsymbol{p}^\star = \left[p_1^\star, \cdots, p_N^\star\right]^T$ and $\boldsymbol{\alpha}^\star = \left[\alpha_1^\star, \cdots, \alpha_N^\star\right]^T$ and denote by $\left(\boldsymbol{p}^\star, \boldsymbol{\alpha}^\star\right)$ the solution (if exists) to the formulated non-cooperative game, which is well-known as the NE~\cite{han2012game}. A NE of the formulated game is a feasible profile $\left(\boldsymbol{p}^\star, \boldsymbol{\alpha}^\star\right)$ that satisfies
\begin{equation}
\begin{cases}
p_n^\star \leq p_n\\
\alpha_n^\star =f_n \left(p_n^\star, \boldsymbol{p}_{-n}^\star\right)
\end{cases}
,~ \forall p_n \in \mathcal{P}_n\left(\boldsymbol{p}_{-n}^\star\right),~\forall n \in \mathcal{N}.
\end{equation}

\section{Existence and Uniqueness of the NE}

In this section, we first analyze the best response strategy of each source-destination pair by solving the optimization problem (\ref{Prob.1}). Then, we derive a sufficient and necessary condition that guarantees the existence and uniqueness for the NE of the formulated game.

\subsection{The Best Response Strategy}
Now, we calculate the best response strategy $(p_n^\star, \alpha_n^\star)$ for the $n$th pair by solving the optimization problem (\ref{Prob.1}), which is described in the following proposition,
\begin{proposition}
Given $\boldsymbol{p}_{-n}$, the optimal response strategy of the $n$th pair can be expressed as
\begin{equation}\label{Eq.BRS}
\begin{cases}
p_n^\star =\frac{ -  X_n + Y_n + \gamma_n X_n + \gamma_n \sigma_n^2  + \sqrt{\Delta_n}}{2G_{nn}},\\
\alpha_n^\star  = \frac{X_n + Y_n + \gamma_n X_n +\gamma_n \sigma_n^2 - \sqrt{\Delta_n}}{2(X_n + \gamma_n X_n)},
\end{cases}
\end{equation}
where $X_n = \sum_{m \in \mathcal{N}/\{n\}} p_m G_{mn} + \delta_n^2$, $Y_n = \frac{\mathcal{E}_n}{\eta}$ and $\Delta _n = \left( X_n - Y_n  +  \gamma_n X_n  +  \gamma_n \sigma_n^2 \right)^2 + 4 \gamma_n Y_n \sigma_n^2$.
\end{proposition}
\begin{proof}
See Appendix A.
\end{proof}

The best response strategy given in Proposition 1 confirms our discussion that the choice of $\alpha_m^\star$, $\forall m\in\mathcal{N} / \{n\}$, of the others will not affect the decision of the $n$th pair. Therefore, the competitive interaction among the players is actually proceeded by adjusting their own power allocation strategy. We thus define the best response function at the $n$th pair as
\begin{equation}\label{Eq.BR}
\begin{split}
p_n^\star = \mathcal{B}_n\left(\boldsymbol{p}_{-n}\right) = \frac{ -  X_n + Y_n + \gamma_n X_n + \gamma_n \sigma_n^2  + \sqrt{\Delta_n}}{2G_{nn}}, ~\forall n \in\mathcal{N}.
\end{split}
\end{equation}
Besides, based on (\ref{Eq.BRS}), we have
\begin{equation}
\begin{split}
f_n \left(p_n, \boldsymbol{p}_{-n}\right) =  \frac{X_n + Y_n + \gamma_n X_n +\gamma_n \sigma_n^2 - \sqrt{\Delta_n}}{2(X_n + \gamma_n X_n)}.
\end{split}
\end{equation}
The NE of the non-cooperative game can now be redefined as
\begin{equation}\label{NE}
\begin{cases}
p_n^\star = \mathcal{B}_n\left(\boldsymbol{p}_{-n}^\star\right)\\
\alpha_n^\star = f_n \left(p_n^\star, \boldsymbol{p}_{-n}^\star\right)
\end{cases}
,~\forall n \in\mathcal{N},
\end{equation}
which can be readily achieved with the well-known \emph{best-response dynamics} \cite{han2012game} if its existence and uniqueness are guaranteed. {It is worth mentioning that the adopted best-response dynamics has a low implementation and computation complexity to achieve the NE of the formulated game in a fully distributed manner. In particular, each link only needs to measure its own channel gain (i.e., $G_{nn}$) and the power of the interference from all other links and the antenna noise (i.e., $X_n$ defined in Proposition 1), which could be realized by equipping a radio scene analyzer at each destination as in \cite{Shi_TWC_2009}. With these local information, each link can easily compute its best response strategy  based on (\ref{Eq.BRS}) and the NE of the formulated non-cooperative game can be readily achieved by the best-response dynamics.}

Note that the NE of the formulated game does not always exist as the EH and SINR constraints of all source-destination pairs cannot be simultaneously satisfied for some special cases. Motivated by this and inspired by the existing literature, we derive a sufficient and necessary condition in the following subsections to guarantee the existence and uniqueness of the NE.

{
\subsection{Existence of NE}
From the structure of the optimization problem (\ref{Prob.1}) at each source-destination pair, we can observe that the NE of the formulated game is existent only when the conditions (C3) and (C4) for all pairs can be met at the same time. Starting from this observation, we obtain the following proposition regarding a sufficient and necessary condition for the existence of NE:
\begin{proposition}
Define a square matrix $\boldsymbol{\Omega} \in \mathbb{R}^{N\times N}$ as
\begin{equation}\label{Eq.Ome}
\left[\boldsymbol{\Omega}\right]_{n,m} =
\begin{cases}
0, &\mbox{if ~$m = n$},\\
 \frac{G_{mn} \gamma_n}{G_{nn}} , &\mbox{if ~$m \neq n$}.
\end{cases}
\end{equation}
The NE of the formulated game exists if and only if (iff) the spectral radius of $\boldsymbol{\Omega}$, denoted by $\rho(\boldsymbol{\Omega})$, is less than one, where the spectral radius of the matrix $\boldsymbol{\Omega}$ is defined as the largest absolute eigenvalue of $\boldsymbol{\Omega}$, i.e., $\rho(\boldsymbol{\Omega}) \triangleq \max \limits_{i} \left(|\lambda_i|\right)$ with $\lambda_i$'s representing eigenvalues of $\boldsymbol{\Omega}$ \cite{hornmatrix}.
\end{proposition}
\begin{proof}
To proceed, we first rewrite all SINR constraints (i.e., (C3)'s) in a matrix form. After some algebra manipulations, we get $\left( {{\bf{I}} - {\bf{\Omega }}} \right) {\pmb p} > \pmb u$ with ${\pmb p} > \pmb 0$ (component-wise), where $\bf{I}$ is the identity matrix, ${\pmb p} = \left[p_1, \cdots, p_N\right]^T$,
\[\pmb u = {\left[ {\frac{{{\gamma _1}\left( {\delta _1^2 + \frac{{\sigma _1^2}}{{1 - {\alpha _1}}}} \right)}}{{{G_{11}}}},
\ldots ,\frac{{{\gamma _N}\left( {\delta _N^2 + \frac{{\sigma _N^2}}{{1 - {\alpha _N}}}} \right)}}{{{G_{NN}}}}} \right]^T},\]
and the matrix $\boldsymbol{\Omega}$ can be expressed as (\ref{Eq.Ome}). With reference to \cite{Bambos_Toward_1998} (i.e., Perron-Frobenious theorem), the existence of a feasible power profile ${\pmb p} > \pmb 0$ satisfying $\left( {{\bf{I}} - {\bf{\Omega }}} \right) {\pmb p} > \pmb u$ is equivalent to $\rho(\boldsymbol{\Omega})< 1$.

We are now rewriting all EH constraints constraints in a matrix form and have
${\bf{G}}{\pmb p} > {\pmb {\mathcal E}}$, where $\bf{G}$ is a $N\times N$ matrix with $\left[{\bf{G}}\right]_{n,m} = \alpha_nG_{mn}$ and the vector ${\pmb {\mathcal E}} = \left[{\mathcal E}_1/\eta,\ldots,{\mathcal E}_N/\eta\right]^T$. It is readily to verify that if there exists a feasible power profile ${\pmb p} > \pmb 0$ satisfying the SINR constraint inequality $\left( {{\bf{I}} - {\bf{\Omega }}} \right) {\pmb p} > \pmb u$, we can always guarantee that the energy constraint inequality ${\bf{G}}{\pmb p} > {\pmb {\mathcal E}}$ is also satisfied by scaling up the feasible profile ${\pmb p}$. Thus, the overall sufficient and necessary condition for the existence of NE is $\rho(\boldsymbol{\Omega})< 1$, which completes the proof.
\end{proof}

It is interesting to observe from Proposition 2 that the condition for the existence of NE only depends on the channel gains and SINR thresholds. Also, the derived condition is actually the same as that for the power minimization problem in conventional IFCs with only SINR constraints, which was given and proved in \cite{nguyen2011multiuser}. This indicates that the introduction of EH constraints does not lead to a stricter condition for the existence of the~NE.

Since a condition is needed to guarantee the existence of NE, there is nonzero probability that the derived condition is not satisfied in certain cases. In these cases, admission control schemes should be adopted to opt out a few source-destination pairs to guarantee the existence of NE among the remaining active pairs. Optimal strategies to opt several pairs out the formulated game in the rounds when the derived condition is not satisfied would be quite difficult to achieve, and are actually beyond the scope of this paper. Moreover, this admission control issue is actually not specific to the adopted game-theoretic approach but rather arises regardless of the method used. The interested reader is referred to \cite{Wu_TVT_2001,Berggren_TVT_2001} and references therein.
}
\subsection{Uniqueness of NE} \label{SC.}
The best-response dynamics, i.e., the iterative behaviors of each player with its best response strategy, can be regarded as a mapping process \cite{bertsekas1989parallel}. A sufficient condition to guarantee the uniqueness of the NE is equivalent to the condition to guarantee the mapping process is a contraction, which then implies that the mapping has a unique fixed point. We then refer to the contraction mapping theorem~\cite{bertsekas1989parallel}, i.e.,
{
\begin{theorem}
A mapping $T$: $\mathcal{X} \rightarrow \mathcal{X}$ ($\mathcal{X}$ is a closed subset of the real number set $\mathbb{R}^N$) is a contraction if and only if $|| T(x) - f(y)|| \leq \beta ||x - y|| $, $\forall x, y \in \mathcal{X}$ and $\beta \in [0,1)$. Moreover, the mapping $T$ has a unique fixed point if it is a contraction.
\end{theorem}
}

Note that the well-known contraction mapping theorem has been widely adopted in open literature due to its following advantages \cite{hogan2009new}. First, it can guarantee the uniqueness of the NE. Second, it does not require that the space being mapped onto itself is convex or bounded. Third, it has the inherent convergence property. Additionally, it is also pointed out in \cite{hogan2009new} that the contraction mapping theorem provides a quite general condition for the mapping to have a unique fixed~point. Then, we can have the following proposition regarding the uniqueness of the~NE:
{
\begin{proposition}
If the NE of the formulated non-cooperative game exists (i.e., the condition $\rho(\boldsymbol{\Omega})< 1$ in Proposition 2 is satisfied), then the game has a unique NE.
\end{proposition}
}
\begin{proof}
See Appendix B.
\end{proof}

{
Till now, we can claim that the derived condition (i.e., $\rho(\boldsymbol{\Omega})< 1$) is actually a sufficient and necessary condition for the existence and uniqueness of the NE of the formulated game. Furthermore, if the condition is satisfied, the best response dynamics is guaranteed to achieve the unique NE in a distributed manner.}

\section{Numerical Results}

In this section, we provide some numerical results to validate the above theoretical analysis. We denote by $d_{mn}$ the inter-link distance between the $m$th source and the $n$th destination, $\forall m \in \mathcal{N} /\{n\}$, and denote by $d_{nn}$ the inner-link distance between the $n$th source and the $n$th destination. All channels are assumed to experience quasi-static flat Rayleigh fading. We adopt a channel model with $\mathbb{E}\left[G_{mn}\right] = 10^{-3}({d_{mn}})^{-\zeta}$ and $\mathbb{E}\left[G_{nn}\right] = 10^{-3}({d_{nn}})^{-\zeta}$, where $\zeta \in [2,5]$ is the path-loss factor and a $30$dB average signal power attenuation is assumed with reference distance of $1$m. In all simulations, we set $d_{nn} = 5$m and $d_{mn} = 10$m unless otherwise stated. Besides, we have $\eta = 0.5$, $\zeta = 3$, $\delta_n^2 = -60$dBm, and $\sigma_n^2 = -50$dBm, $\forall n\in\mathcal{N}$.

\begin{figure}
\centering \scalebox{0.7}{\includegraphics{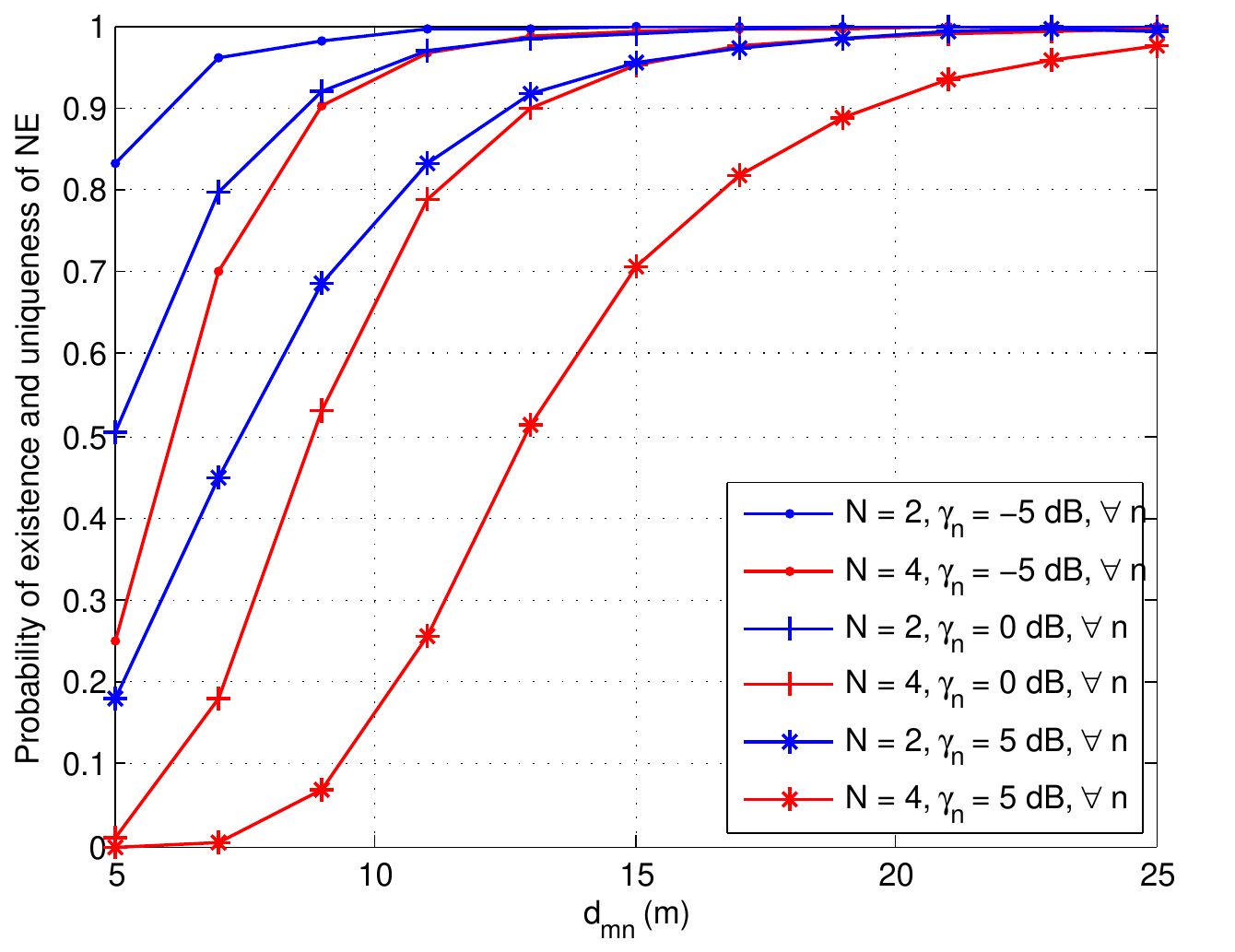}}
\caption{{Probability of the existence and uniqueness of the NE versus the inter-link distance $d_{mn}$ with different values of $N$ and $\gamma_n$, where $d_{nn} = 5$m.} }
\label{fig.exist}
\end{figure}
{Since the existence and uniqueness of NE of the formulated game depend on the channel gains $\left\{G_{mn}\right\}$ (See Proposition 2-3), there is a nonzero probability that the sufficient and necessary condition $\rho(\boldsymbol{\Omega}) < 1$ is not satisfied for a certain set of channel realization and SINR thresholds. To quantify how frequently the condition $\rho(\boldsymbol{\Omega}) < 1$ holds, we perform Monte Carlo simulations to evaluate the probability of existence and uniqueness of NE (i.e., the condition $\rho(\boldsymbol{\Omega}) < 1$ is satisfied) for several different scenarios. Specifically, in Fig.~\ref{fig.exist} we plot the curves of this probability versus the inter-link distance with different numbers of links and SINR thresholds. We can observe from this figure that, the probability of $\rho(\boldsymbol{\Omega}) < 1$ grows quickly as the inter-link distance increases and can approach to one when the inter-link distance is large enough for all simulated cases. Moreover, this probability also improves when either the number of links or SINR thresholds decreases. Both of these observations are caused by a decrease of inter-link interference. As discussed in the previous section, when the condition $\rho(\boldsymbol{\Omega}) < 1$ is not satisfied, admission control schemes should be adopted to opt out a few source-destination pairs to guarantee the existence (also uniqueness) of NE among the remaining active pairs. For simplicity, in the sequel we consider the scenario that appropriate admission control schemes have been carried out at the beginning of each transmission block such that the NE of the formulated game is guaranteed to exist among the active source-destination pairs, i.e., the condition $\rho(\boldsymbol{\Omega}) < 1$ is assumed to be~satisfied.
}

\begin{figure}
\centering
 \subfigure[Convergence of $p_n^\star$.]
  {\scalebox{0.55}{\includegraphics {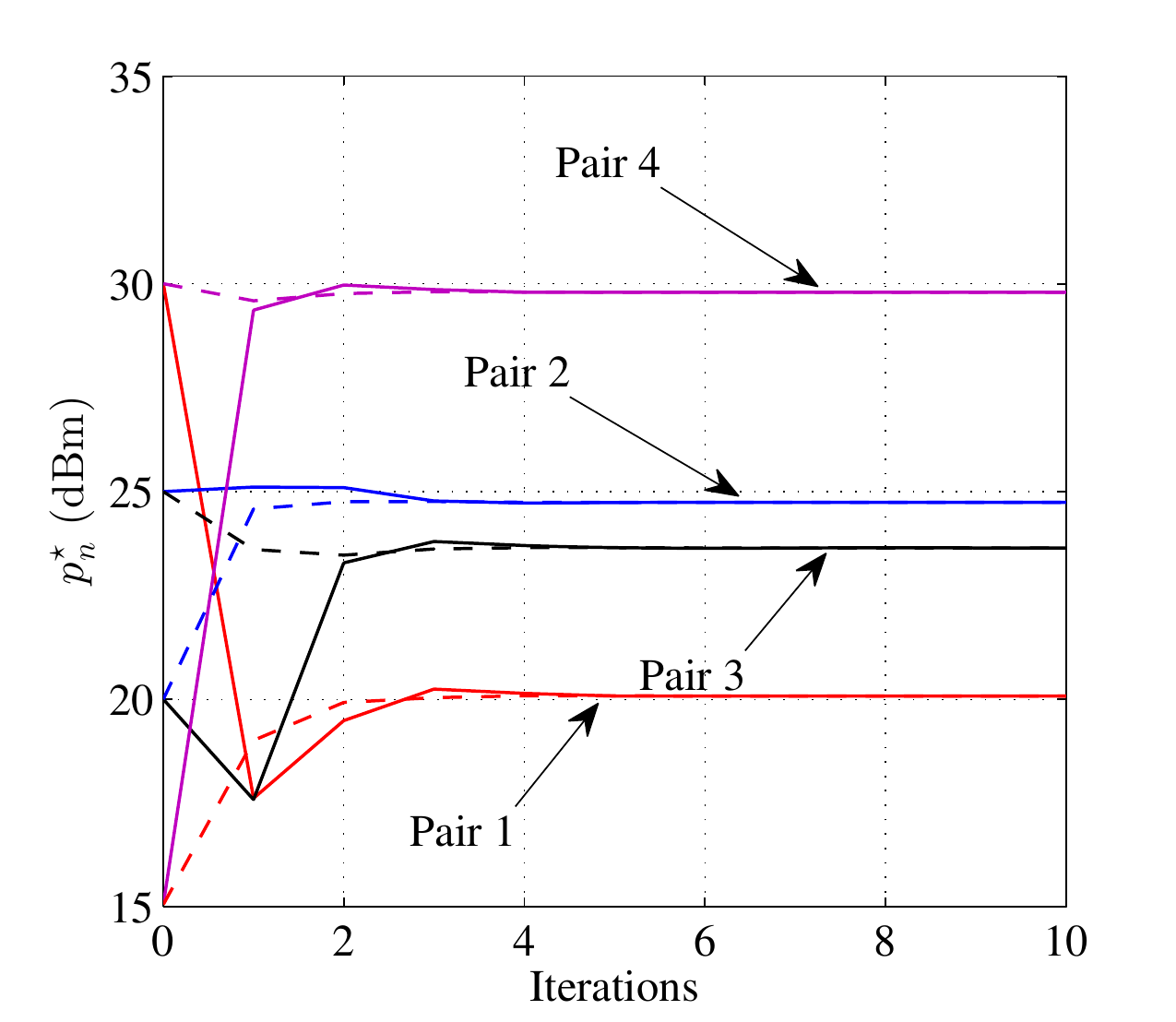}
  \label{Fig.p}}}
\hfil
 \subfigure[Convergence of $\alpha_n^\star$.]
  {\scalebox{0.55}{\includegraphics {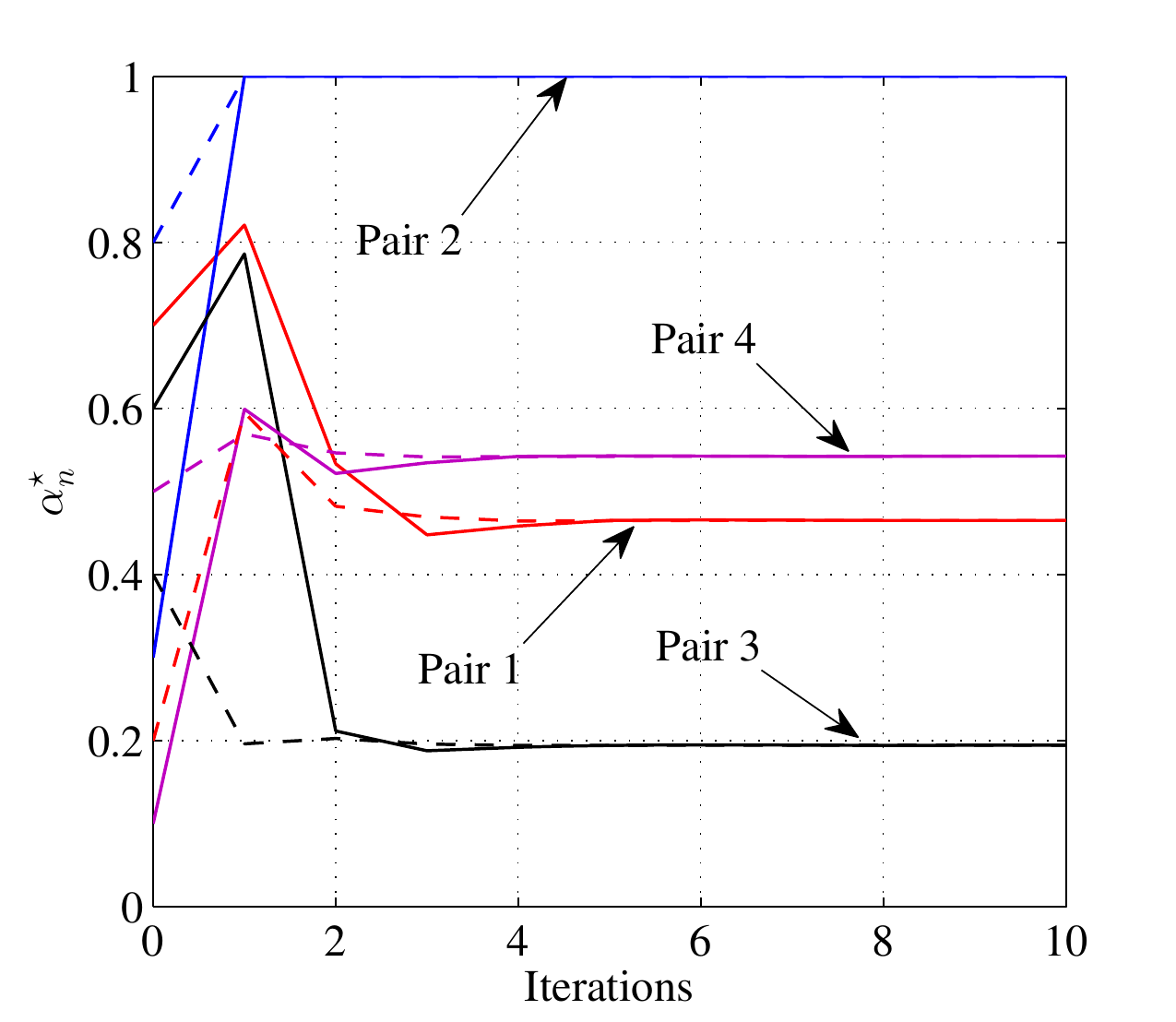}
\label{Fig.alpha}}}
\caption{The convergence of best-response dynamics in the formulated non-cooperative game with four source-destination pairs starting from two different sets of initial points, which are distinguished by solid lines and dash lines.}
\label{fig2}
\end{figure}
Next, we demonstrate the convergence of best-response dynamics in the formulated non-cooperative game for a four-pair setup with one randomly generated channel realization as well as SINR and EH constraints. The pairs 1-4 are assumed with SINR constraints $0$dB, $0$dB, $10$dB, $10$dB, and EH constraints $-20$dBm, $-10$dBm, $-20$dBm , $-10$dBm, respectively. Fig. \ref{Fig.p} and Fig. \ref{Fig.alpha} show the transmit power and the power splitting ratio of each pair versus the number of iterations, respectively. Two cases starting from two random sets of initial points are presented, which are distinguished by dash and solid lines, respectively. It can be observed from this figure that both $p_n^\star$ and $\alpha_n^\star$ can converge quickly to the same stationary values (i.e., the NE) from different starting points. Also, as the minimum $p_n^\star$ is about $20$dBm, the signal-to-noise ratio (SNR) region is over $70$dB ($\sigma_n^2 = -50$dBm), which is practical for energy harvesting devices. Note that due to the space limitation, we only show results in Fig. \ref{fig2} for one random realization of channel gains and the constraints, although similar results can also be shown for other realizations. This validates the effectiveness of the derived sufficient and necessary condition of the NE. { At last, it is worth pointing out that the optimal $\alpha_2$ is shown to be closed to one in Fig. \ref{Fig.alpha}. This phenomenon corresponds to the scenario that the inter-link interference suffered by pair 2 is very weak such that it only needs a sufficiently small transmit power to meet its SINR constraint. In such a case, the transmit power of this pair is mainly dominated by its EH constraint and thus the optimal value of $\alpha_2$ for this link should approach to one to align with the objective of minimizing the source transmit power.
}
\begin{figure}
\centering \scalebox{0.7}{\includegraphics{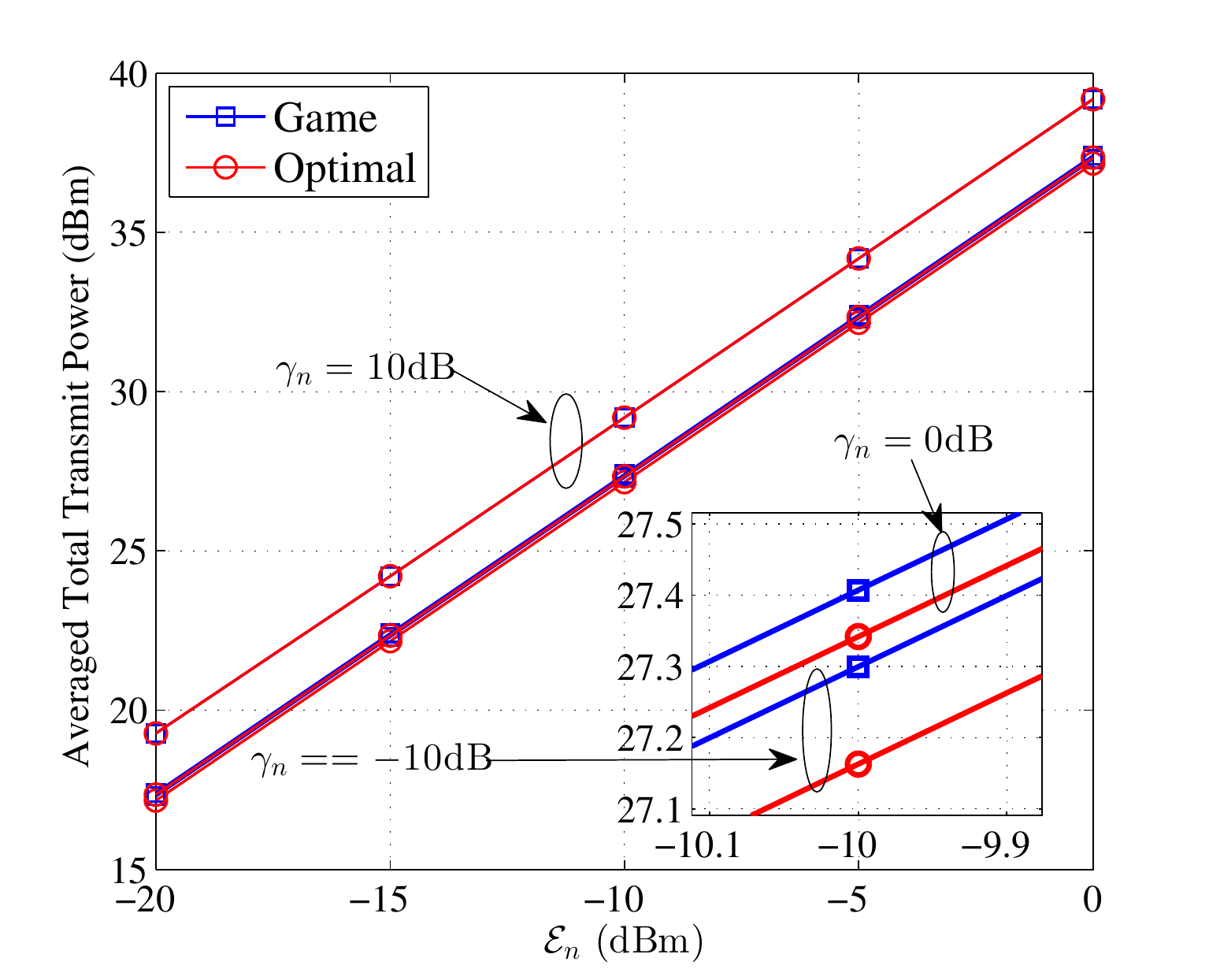}}
\caption{Averaged total transmit power versus the EH constraint with different SINR constraints in a two-pair network.}
\label{fig.sum}
\end{figure}

Fig. \ref{fig.sum} illustrates the averaged total transmit power versus the EH constraints with different SINR constraints in a two-pair network setup. The averaged total transmit power obtained by the proposed game-theoretic approach is compared with an optimal strategy calculated via an exhaustive search. The optimal strategy is conducted based on the assumption that all source-destination pairs cooperate with each other to minimize the averaged total transmit power. We note that the optimal scheme here is actually consistent with the problems considered in \cite{timotheou2013beamforming} and \cite{shi2014joint}. The SINR and EH constraints for two pairs are assumed to be the same, respectively. Each curve is obtained by averaging over $10^4$ independent channel realizations. It can be observed from Fig. \ref{fig.sum} that, the averaged total transmit power of the proposed game-theoretic method can closely match that of the optimal results. Note that the proposed strategy cannot perfectly coincide with the optimal strategy due to the rational and selfish behaviors of source-destination pairs in our game formulation. In addition, it is also shown that with the increasing of $\gamma_n$ or $\mathcal{E}_n$, the averaged total transmit power increases accordingly to meet the higher SINR and EH constraints.

\section{Conclusion}
In this paper, we developed a game-theoretic framework to tackle the distributed joint power and power splitting ratio problem in a IFC with SWIPT. We formulated a non-cooperative game for the considered system, where each source-destination pair is modeled as a selfish and rational player who aims to minimize its own transmit power under both SINR and EH constraints. We derived the best response strategy of each player and thus the NE can be obtained iteratively. A sufficient and necessary condition to guarantee the existence and uniqueness of the NE was also provided. The numerical results validated the derived condition and showed that the performance of the proposed game-theoretic approach can closely match the optimal strategy on averaged transmit power under various SINR and EH constraints.
\section*{Appendix A\\Proof of Proposition 1}
First of all, we could notice that the constraint (C3) and (C4) should hold with equality at the optimal solution; otherwise, we can always reduce the power of $p_n$. Therefore, we can solve the following equations to achieve the optimal solution of the formulated problem in (\ref{Prob.1})
 \begin{equation}\label{Eq.13}
\begin{cases}
\gamma_n - \frac{\left(1- \alpha_n\right) p_n G_{nn}}{\left(1-\alpha_n\right)\left(  \sum_{m \in \mathcal{N}/\{n\}} p_m G_{mn} + \delta_n^2 \right) + \sigma_n^2} = 0,\\
\mathcal{E}_n - \eta \alpha_n  \sum_{m\in\mathcal{N}}  p_m G_{mn}  = 0.
\end{cases}
\end{equation}
Let $X_n = \sum_{m \in \mathcal{N}/\{n\}} p_m G_{mn} + \delta_n^2 $ and $Y_n = \frac{\mathcal{E}_n}{\eta}$, we have
\begin{equation}
\begin{cases}
\frac{\left(1- \alpha_n\right) p_n G_{nn}}{\left(1-\alpha_n\right)X_n + \sigma^2} = \gamma_n ,\\
 \alpha_n \left(p_n G_{nn} + X_n \right)  = Y_n.
\end{cases}
\end{equation}
By solving $\alpha_n$ first, we can readily obtain
\begin{equation}
\begin{split}
\alpha_n = \frac{X_n+Y_n + \gamma_n X_n +\gamma_n \sigma_n^2 \pm \sqrt{\Delta_n}}{2(X_n + \gamma_n X_n)},
\end{split}
\end{equation}
where $\Delta_n \triangleq (X_n - Y_n + \gamma_n X_n + \gamma_n \sigma_n^2)^2 + 4\gamma_n Y_n \sigma_n^2 > 0$.

Recall that the optimization problem (\ref{Prob.1}) always has a feasible solution and $0\leq \alpha_n \leq 1$. It is easy to check that
\begin{equation}
\begin{split}
\alpha_n & = \frac{X_n+Y_n + \gamma_n X_n +\gamma_n \sigma_n^2 + \sqrt{\Delta_n}}{2(X_n + \gamma_n X_n)}\\
& > \frac{X_n+Y_n + \gamma_n X_n +\gamma_n \sigma_n^2 + \left| X_n - Y_n + \gamma_n X_n + \gamma_n \sigma_n^2\right|}{2(X_n + \gamma_n X_n)} \\
&  > 1,
\end{split}
\end{equation}
which is invalid. We thus obtain
\begin{equation}
\alpha_n^\star  = \frac{X_n + Y_n + \gamma_n X_n +\gamma_n \sigma_n^2 - \sqrt{\Delta_n}}{2(X_n + \gamma_n X_n)},
\end{equation}
and with some algebra manipulations, we have
\begin{equation}
p_n^\star  =\frac{ -  X_n + Y_n + \gamma_n X_n + \gamma_n \sigma_n^2  + \sqrt{\Delta_n}}{2G_{nn}},
\end{equation}
which completes the proof. 

\section*{Appendix B\\Proof of Proposition 2}

We first set $\mathcal{T}_n(\boldsymbol{p}) = \mathcal{B}_n(\boldsymbol{p}_{-n})$ and $\boldsymbol{\mathcal{T}}(\boldsymbol{p}) = \left(\mathcal{T}_n (\boldsymbol{p}) \right)_{n \in\mathcal{N}}$. Then, the proof of this proposition follows if the sufficient and necessary condition derived in Proposition 2 can guarantee that  the mapping $\boldsymbol{\mathcal{T}}(\boldsymbol{p})$ is a contraction mapping. To this end, we define that $\varphi_{\mathcal{T}_n(\boldsymbol{p})} = |\mathcal{T}_n(\boldsymbol{p}) -\mathcal{T}_n(\boldsymbol{p}') | $ and $\varphi_{n} = | p_n - p_n' |$, $\forall \boldsymbol{p},\boldsymbol{p}' \geq 0 $.
Recall that  $X_n = \sum_{m \in \mathcal{N}/\{n\}} p_m G_{mn} +\delta_n^2 $ and $\Delta _n = \left( X_n - Y_n  +  \gamma_n X_n  +  \gamma_n \sigma_n^2 \right)^2 + 4 \gamma_n Y_n \sigma_n^2$, we have $X_n' = \sum_{m \in \mathcal{N}/\{n\}} p_m' G_{mn} + \delta_n^2 $ and $\Delta _n' = \left( X_n' - Y_n  +  \gamma_n X_n'  +  \gamma_n \sigma_n^2 \right)^2 + 4 \gamma_n Y_n \sigma_n^2$. Then, we have
\begin{equation}\label{eq:111}
\begin{split}
 \varphi_{\mathcal{T}_n(\boldsymbol{p})} = &|\mathcal{T}_n(\boldsymbol{p}) -\mathcal{T}_n(\boldsymbol{p}') |\\
   =& \left | \frac{ { - X_n   + \gamma_n X_n   + \sqrt{ \Delta_n}}   + X_n' - \gamma_n X_n'  - \sqrt{ \Delta_n'} }{2 G_{nn}} \right|\\
= &\left | \frac{1}{2 G_{nn}} \left[  (\gamma_n -1)(X_n - X_n')  + \frac{\Delta_n - \Delta_n'} {\sqrt{ \Delta_n } + \sqrt{ \Delta_n'} } \right]\right|\\
\end{split}
\end{equation}
{
We can further simplify the term $\Delta_n - \Delta_n'$ in (\ref{eq:111}) as follows
\begin{equation}\label{eq:222}
\begin{split}
\Delta_n - \Delta_n'=& \left( X_n - Y_n  +  \gamma_n X_n  +  \gamma_n \sigma_n^2 \right)^2 - \left( X_n' - Y_n  +  \gamma_n X_n'  +  \gamma_n \sigma_n^2 \right)^2\\
=& \left[\left( X_n - Y_n  +  \gamma_n X_n  +  \gamma_n \sigma_n^2 \right) + \left( X_n' - Y_n  +  \gamma_n X_n'  +  \gamma_n \sigma_n^2 \right)\right] \times \\
&\left[\left( X_n - Y_n  +  \gamma_n X_n  +  \gamma_n \sigma_n^2 \right) - \left( X_n' - Y_n  +  \gamma_n X_n'  +  \gamma_n \sigma_n^2 \right)\right]\\
=& \left[\left( X_n - Y_n  +  \gamma_n X_n  +  \gamma_n \sigma_n^2 \right) + \left( X_n' - Y_n  +  \gamma_n X_n'  +  \gamma_n \sigma_n^2 \right)\right] \left[\left(\gamma_n + 1\right) \left( X_n - X_n'\right) \right].
\end{split}
\end{equation}
Substituting (\ref{eq:222}) into (\ref{eq:111}), we have
\begin{equation}\label{eq:333}
\begin{split}
 \varphi_{\mathcal{T}_n(\boldsymbol{p})} = \left | \frac{X_n - X_n' }{2 G_{nn}} \left[  \gamma_n -1  + Z_n (\gamma_n + 1) \right]\right|,
\end{split}
\end{equation}
where $Z_n$ is defined as
\begin{equation}\label{Eq.Z}
\begin{split}
Z_n \triangleq  \frac{\left(X_n - Y_n  +  \gamma_n X_n  +  \gamma_n \sigma_n^2\right) +  \left(X_n' - Y_n  +  \gamma_n X_n'  +  \gamma_n \sigma_n^2\right) } {\sqrt{ \Delta_n } + \sqrt{ \Delta_n'} }.
\end{split}
\end{equation}
By realizing that $\sqrt{ \Delta_n } > \left|X_n - Y_n  +  \gamma_n X_n  +  \gamma_n \sigma_n^2 \right| >0$ and $\sqrt{ \Delta_n' } > \left|X_n' - Y_n  +  \gamma_n X_n'  +  \gamma_n \sigma_n^2 \right|>0$, $\forall \gamma_n, \mathcal{E}_n, \sigma_n^2 > 0$, we have the following inequality regarding $Z_n$
\begin{equation}
\begin{split}
|Z_n| &< \frac{\left|\left(X_n - Y_n  +  \gamma_n X_n  +  \gamma_n \sigma_n^2\right) +  \left(X_n' - Y_n  +  \gamma_n X_n'  +  \gamma_n \sigma_n^2\right) \right|} {\left|X_n - Y_n  +  \gamma_n X_n  +  \gamma_n \sigma_n^2 \right| + \left|X_n' - Y_n  +  \gamma_n X_n'  +  \gamma_n \sigma_n^2 \right| }\\
&\le \frac{\left|X_n - Y_n  +  \gamma_n X_n  +  \gamma_n \sigma_n^2\right| +  \left|X_n' - Y_n  +  \gamma_n X_n'  +  \gamma_n \sigma_n^2\right| } {\left|X_n - Y_n  +  \gamma_n X_n  +  \gamma_n \sigma_n^2 \right| + \left|X_n' - Y_n  +  \gamma_n X_n'  +  \gamma_n \sigma_n^2 \right| }\\
& = 1,
\end{split}
\end{equation}
where the second inequality holds by following the fact that $\left|a+b\right|\le \left|a\right| + \left|b\right|$. With $|Z_n| < 1$ proved above, we can further simplify (\ref{eq:333}) as follows
\begin{equation}
\begin{split}
\varphi_{\mathcal{T}_n(\boldsymbol{p})} = & \left | \frac{X_n - X_n' }{2 G_{nn}} \left[ \gamma_n -1 + Z_n (\gamma_n + 1)  \right]\right |
<  \sum_{m \in\mathcal{N}/{n}} \left| \frac{G_{mn} \gamma_n}{G_{nn}} \right| \left |p_m - p_m' \right |
 = \sum_{m \in\mathcal{N}/{n}} \frac{G_{mn} \gamma_n}{G_{nn}} \varphi_{m} .
\end{split}
\end{equation}
}
We now define the vectors $\boldsymbol{\varphi}_{\mathcal{T}}  = \left[ \varphi_{\mathcal{T}_1(\boldsymbol{p})}, \cdots, \varphi_{\mathcal{T}_N(\boldsymbol{p})}\right]^T$, $\boldsymbol{\varphi} = \left[\varphi_{1}, \cdots, \varphi_{N} \right]^T$ and define the square matrix $\boldsymbol{\Omega} \in \mathbb{R}^{N \times N}$ as in (\ref{Eq.Ome}).

We thus have $\boldsymbol{\varphi}_{\mathcal{T}}  < \boldsymbol{\Omega} \boldsymbol{\varphi}$. With reference to \cite{nguyen2011multiuser}, we can obtain
\begin{equation}\label{mapping}
\begin{split}
\| \boldsymbol{\varphi}_{\mathcal{T}} \|_{2,block}^{\boldsymbol{w}} = \|\boldsymbol{\mathcal{T}}(\boldsymbol{p}) - \boldsymbol{\mathcal{T}}(\boldsymbol{p}')\|_{2,block}^{\boldsymbol{w}} < \|\boldsymbol{\Omega}\|_{\infty, mat}^{\boldsymbol{w}} \|\boldsymbol{p} -  \boldsymbol{p}'\|_{2, block}^{\boldsymbol{w}},
\end{split}
\end{equation}
where $\|\boldsymbol{x} \|_{2,block}^{\boldsymbol{w}}$ is the block-maximum norm of a vector $\boldsymbol{x}$ for a positive vector $\boldsymbol{w}$ \cite{bertsekas1989parallel} and $\|\cdot\|_{\infty, mat}^{\boldsymbol{w}}$ is the induced $\infty$-norm for matrix \cite{hornmatrix}. Thus, if $\|\boldsymbol{\Omega}\|_{\infty, mat}^{\boldsymbol{w}} < 1$, the NE is guaranteed to be unique because the mapping in (\ref{mapping}) is a contraction. As $\boldsymbol{\Omega}$ is a nonnegative matrix, we have $\|\boldsymbol{\Omega}\|_{\infty, mat}^{\boldsymbol{w}} < 1 \Leftrightarrow \rho(\boldsymbol{\Omega}) < 1$ \cite{bertsekas1989parallel},
which is guranteed when the NE of the formulated game exists (see Proposition 2). This completes the proof. %

\bibliographystyle{IEEEtran}
\bibliography{SWIPT-IFC}

\end{document}